\pdfoutput=1
\documentclass[11pt]{article}
\usepackage[letterpaper, left=1in, right=1in, top=1in,bottom=1in]{geometry}

\usepackage[boxed]{algorithm2e}

\SetAlFnt{\small}
\SetAlCapFnt{\small}
\SetAlCapNameFnt{\small}
\SetAlCapHSkip{0pt}
\IncMargin{-\parindent}

\usepackage{color}              
\usepackage[suppress]{color-edits}
\addauthor{et}{black}
\addauthor{pj}{black}
\addauthor{etn}{blue}
\usepackage{graphicx}
\usepackage{amsmath,amssymb,amsthm}
\usepackage{geometry}
\usepackage{tikz}
\usetikzlibrary{chains,arrows}
\usepackage{wrapfig}

\theoremstyle{plain}
\newtheorem{proposition}{Proposition}
\newtheorem{theorem}{Theorem}[section]
\newtheorem{lemma}[theorem]{Lemma}
\newtheorem{corollary}[theorem]{Corollary}
\theoremstyle{definition}

\def\squareforqed{\hbox{\rule{2.5mm}{2.5mm}}}

\def\QED{\ifmmode\squareforqed 
  \else{\nobreak\hfil   
    \penalty50                 
    \hskip1em                  
    \null                      
    \nobreak                   
    \hfil                      
    \squareforqed              
    \parfillskip=0pt           
    \finalhyphendemerits=0     
    \endgraf}                  
  \fi}

\def\blksquare{\rule{2mm}{2mm}}
\def\qedsymbol{\blksquare}
\newcommand{\bg}[1]{\medskip\noindent{\bf #1}}
\newcommand{\ed}{{\hfill\qedsymbol}\medskip}

\newtheorem{thm}{Theorem}

\newcommand{\SubmissionOmit}[1]{} 

\newcommand{\fullversion}[1]{}

\usepackage[boxed]{algorithm2e}

\usepackage{amsmath,amsfonts,amssymb,bbm}

\usepackage{dsfont}
\usepackage{ifthen}
\usepackage{commath}
\usepackage{graphicx}
\usepackage{enumerate}
\usepackage{color}
\usepackage{tikz}
\usepackage{hyperref}
\usepackage[capitalize]{cleveref}

\tikzstyle{marketstyle}=[x=3cm,y=2cm,node font=\small]%
\tikzstyle{every edge quotes}=[auto=false,fill=white,opacity=.95,text opacity=1,pos=.25,inner sep=1pt]%
\tikzstyle{every node}+=[inner sep=2.5pt]%

\SetArgSty{textrm}  
\SetAlFnt{\small}
\SetAlCapFnt{\small}
\SetAlCapNameFnt{\small}
\SetAlCapHSkip{0pt}
\IncMargin{-\parindent}

\DeclareMathAlphabet{\mathbfit}{OML}{cmm}{b}{it}

\newcommand{\suchthat}{\;\ifnum\currentgrouptype=16 \middle\fi|\;}

\newcommand{\myref}[2]{\hyperref[#2]{$#1$\ref*{#2}}}

\newcommand{\argmax}{\operatorname{arg\ max}}

\crefname{property}{Property}{Properties}

\crefname{program}{Program}{Programs}

\begin{document}


\title{Simple and Efficient
Budget Feasible Mechanisms for Monotone Submodular Valuations}

\author{Pooya Jalaly \thanks{
Cornell University, 
Department of Computer Science, 
Ithaca NY 14850 USA. Email: \texttt{jalaly@cs.cornell.edu}. Work supported in part by NSF grant CCF-1563714, ONR grant N00014-08-1-0031, and a Google Research Grant. }
\and
\'{E}va Tardos \thanks{
Cornell University, 
Department of Computer Science, 
Ithaca NY 14850 USA.
Email: \texttt{eva@cs.cornell.edu}. Work supported in part by NSF grant CCF-1563714, ONR grant N00014-08-1-0031, and a Google Research Grant.}
}
\maketitle
\begin{abstract}
We study the problem of a budget limited buyer who wants to buy a set of items, each from a different seller, to maximize her value. The budget feasible mechanism design problem aims to design a mechanism which incentivizes the sellers to truthfully report their cost, and maximizes the buyer's value while guaranteeing that the total payment does not exceed her budget. Such budget feasible mechanisms can model a buyer in a crowdsourcing market interested in recruiting a set of workers (sellers) to accomplish a task for her.

This budget feasible mechanism design problem was introduced by Singer in 2010.  There have been a number of improvements on the approximation guarantee of such mechanisms since then. We consider the general case where the buyer's valuation is a monotone submodular function.
We offer two general frameworks for simple mechanisms, and by combining these frameworks, we significantly improve on the best known results for this problem, while also simplifying the analysis. For example, we improve the approximation guarantee for the general monotone submodular case from 7.91 to 5; and for the case of large markets (where each individual item has negligible value) from 3 to 2.58. More generally, given an $r$ approximation algorithm for the optimization problem (ignoring incentives), our mechanism is a $r+1$ approximation mechanism for large markets, an improvement from $2r^2$. We also provide a similar parameterized mechanism without the large market assumption, where we achieve a $4r+1$ approximation guarantee.
\end{abstract}


\newpage




\section{Introduction}
We study \textit{prior-free budget feasible mechanism design} problem, where a single buyer aims to buy a set of items, each from a different seller. Budget feasible mechanism design 
aims to maximize the value of the buyer, while keeping the total payments bellow the budget. We offer simple and universally truthful mechanisms for this problem, significantly improving previous bounds. This problem was introduced by \cite{singer2010}, and models the problem of crowdsourcing platforms, such as Amazon's Mechanical Turk, where a requester, with a set of tasks at hand, wishes to procure a set of workers to accomplish her tasks. Each worker has a private cost for his service. We offer universally truthful mechanisms with good approximation guarantee for this problem that incentivizes the workers to report their true cost.

We give two very simple parametrized mechanisms, assuming the buyer's valuation is a general monotone (non-decreasing) submodular function. Monotone submodular functions are widely used and general, submodularity capturing the diminishing returns property of adding items. Submodular value functions are the most general class of functions where the optimization problem (without considering incentives) can be solved approximately in polynomial time using a value oracle. We consider this problem in the general case, as well as the special case where each item individually does not have a significant value compared to the optimum. 



\paragraph{\textbf{Our Model}}
\label{Sec:Model}
We consider the problem of a single buyer with a budget $B$ facing a set of multiple sellers $A$. 
We assume that each seller $i \in A$ has a single indivisible item and has a private cost $c_i$, and the buyer has no prior knowledge of the private costs. The utility of a seller for selling her item and receiving payment $p_i$ is $p_i-c_i$. We only study universally truthful mechanisms, i.e. the mechanisms in which sellers truthfully report their costs, and do not have incentive to misreport. Since each seller $i \in A$ only has a single item, we interchangeably use $i$ to denote the seller or his item.  We assume that $v(S)$, the value of the buyer for a subset of items $S \subseteq A$, is a monotone (non-decreasing) submodular function.


The \textit{budget feasibility constraint} enforces the total payments 
to the sellers to never be higher than the budget. The goal of this paper is to design simple, universally truthful, budget feasible mechanisms that maximize the value of the buyer. We compare the performance of our mechanism with the true optimum without computational or incentive limitation, 
the total cost of the items selected needs to be bellow the budget. With this comparison in mind, incentive compatible mechanisms that do not run in polynomial time are also of interest.



\paragraph{\textbf{Our Contribution}}
\label{Sec:Contribution}
We offer two classes of parameterized mechanisms. In Section \ref{Sec:Param1}, we study the class of parameterized \textit{threshold mechanisms} that decide on adding items based on a threshold of the marginal contribution of each item over its cost (bang per buck), using a parameter $\gamma$. In section \ref{Sec:Param2}, we consider another parameterized class, called the \textit{oracle mechanisms}, which adds items in decreasing order of bang per buck, till reaching an $\alpha$ fraction of the true optimum, without considering the budget. In section \ref{Sec:Param} we analyze these two parametrized mechanisms for general monotone submodular valuations. In section \ref{Sec:LargeMarkets} we combine the two mechanisms to get an improved result for large markets. See Table 1 for a summary of our results \etnedit{for the general problem}. \etnedit{In section \ref{Sec:Application} we focus on the application to a problem of markets with heterogenous tasks \cite{chen2011,goel2014}.}
%

\begin{table}[h]
\label{Tab:Summary}
\begin{center}
  \begin{tabular} {|l|c|c|c|c|c|c|}
    \hline
     & Rand & Rand$^*$ & Det$^*$& Det, LM & Rand, Oracle& Det, Oracle, LM\\
     \hline
     Previous work & $7.91 $ & $7.91 $& $8.34 $ & $3 $ & $-$ & $2r^2 $\\
    \hline
    Our results &
     $5$ & $4$& $4.56$ & $2.58$ & $4r$ or $4r+1$ & $1+r$\\
    \hline
  \end{tabular}
\caption{\small The top numbers are the previously known best guarantees, $r\geq 1$ is the approximation ratio of the oracle used by the mechanism, $^*$ indicates that the mechanism has exponential running time. Rand and Det stand for randomized and  deterministic mechanisms, and LM indicates the large market assumption. The $4r$ guarantee requires an additional assumption for the oracle, without the assumption the bound is $4r+1$.}
\end{center}
\vspace{-0.3in}
\end{table}

\begin{itemize}
\item In section \ref{Sec:Param1} we consider threshold mechanism \textsc{Greedy-TM}, and \textsc{Random-TM}, that chooses randomly between the single item of highest value, and the output of the \textsc{Greedy-TM} mechanisms. Most of the mechanisms presented in \cite{singer2010}, \cite{chen2011}, \cite{singla2013}, and some of the mechanisms of \cite{anari2014} are special cases of this framework. We show that for monotone submodular valuations, with the right choice of the parameter $\gamma$, our randomized threshold mechanism is universally truthful, budget feasible and can achieve a 5 approximation of the optimum. This improves on the best previously known bound of 7.91 which is due to \cite{chen2011}. 
\item In section \ref{Sec:Param2}, we give another parameterized class of mechanisms, \textsc{Greedy-OM} and \textsc{Greedy-EOM}, called \textit{oracle mechanisms}, which add items in the bang per bunk order until an $\alpha$ fraction of the optimum value is obtained, for a parameter $\alpha$. The mechanism \textsc{Greedy-EOM} uses the true optimum value, while \textsc{Greedy-OM} uses a polynomial time approximation instead. We show that keeping the total value of the winning set at most a fraction of the optimum, guarantees that the mechanism is budget feasible. The large market oracle mechanisms of \cite{anari2014} are a special case of these mechanisms.

    For the case when the algorithm has access to an oracle computing the true optimum value, we show that with the right choice of $\alpha$, our oracle mechanism \textsc{Random-EOM} is universally truthful, budget feasible and achieves a 4 approximation of the optimum for monotone submodular values, improving the bound of 7.91 of \cite{chen2011}. We also use a derandomization idea, which is similar to that of \cite{chen2011}, to give an (exponential time) oracle mechanism which achieves 4.56 approximation of the optimum, improving the 8.34 bound of \cite{chen2011}.

    The mechanisms \textsc{Random-OM} and \textsc{Greedy-OM}, run in polynomial time, using an $r$-approximation oracle as a subroutine instead of the optimum. We show that with the right choice of $\alpha$, \textsc{Greedy-OM} is universally truthful, budget feasible, and achieves a $4r+1$ approximation of the optimum (which improves to $4r$ when the oracle used is a greedy algorithm).
\item In section \ref{Sec:LargeMarkets}, we combine our two parameterized mechanisms by running both, and outputting the intersection of the two sets. Taking the intersection allows us to use larger values of the parameters $\gamma$ and $\alpha$ and keep the mechanism budget feasible. We show that our simple mechanism is universally truthful. 
    For the right choice of $\alpha$ and $\gamma$, our mechanism is deterministic, truthful, budget feasible and has an approximation guarantee of $1+r$, improving the bound $2r^2$ of \cite{anari2014} (where $r$ is the approximation guarantee of the oracle used). Using the greedy algorithm of \cite{greedy2004} (which was also analyzed in \cite{khuller1999} for linear valuations), the approximation guarantee of our mechanism is $1+\frac{e}{e-1}\simeq 2.58$.
\item In section \ref{Sec:Application}, we show how that our results for submodular valuations can be used for the problem of Crowdsourcing Markets with Heterogeneous Tasks  of \cite{goel2014}. Our mechanism of section \ref{Sec:LargeMarkets}, gives a deterministic truthful and budget feasible mechanism with an $1+\frac{e}{e-1}\approx 2.58$ approximation guarantee for large markets. The previous best result is the randomized truthful (in expectation) mechanism of \cite{goel2014}. We match their guarantee with a deterministic truthful mechanism.
\end{itemize}

\paragraph{\textbf{Related Work}}
\label{Sec:Related}
Prior free budget feasible mechanism design for buying a set of items, each from a different seller, has been introduced by \cite{singer2010}. For monotone submodular valuations, which is the focus of our paper, 
\cite{chen2011} improved the mechanism and its analysis to achieve a 7.91 approximation guarantee, and also derandomized the mechanism to get a deterministic (but exponential time) mechanism with 
an approximation guarantee of 8.34.

\cite{singla2013} considered the problem for an application in community sensing 
and gave a mechanism with a 4.75 approximation guarantee
for large markets.
\cite{anari2014} improved the result of \cite{singla2013} achieving a 3 approximation guarantee for large markets with a polynomial time mechanism and a 2 approximation guarantee with an exponential time mechanism.
\cite{anari2014} also proposed a mechanism that given an $r$ approximation oracle for maximizing monotone submodular functions, has a $2r^2$ approximation guarantee. 

The problem of budget feasible mechanisms for value maximization has also been considered with other valuation functions. For example, for additive valuations the best known mechanism achieves an approximation bound of $2 +\sqrt{2}$ and 3 with a deterministic and randomized mechanisms respectively due to \cite{chen2011}, who also gave a $1+\sqrt{2}$ lower bound for approximation ratio of any truthful budget feasible mechanism in this setting. In large markets with additive valuations, \cite{anari2014} improved these results and gave a budget feasible mechanism with an approximation guarantee of $\frac{e}{e-1}$ with a matching lower bound.

\etnedit{\cite{singer2010} also considered the budget feasible mechanism design problem with matching constraint: the principal is required to select a matching of a bipartite graph, where each individual edge is an agent (with a private cost and a public value). \cite{chen2011} consider the knapsack problem with heterogeneous items, which is the special case of this problem where the bipartite graph is a set of disjoint stars. Their approximation bound mentioned above, of $2 +\sqrt{2}$ and 3 with a deterministic and randomized mechanisms, also extend to this case. \cite{goel2014} considered a variant of the problem motivated by Crowdsourcing Markets, where one side of the graph are agents with private costs, and the other side are tasks, each with a value for the principal, and an edge $(a,t)$ represents that agent $a$ can do task $t$. They give a randomized truthful (in expectation) mechanism with a $1+\frac{e}{e-1}$ approximation guarantee for this problem under the large market assumption. We'll discuss this special case in section \ref{Sec:Application}, and give a deterministic truthful mechanism matching their approximation guarantee.}

Prior free budget feasible mechanisms has also been studied for more general valuation functions. Monotone submodular valuations are the most general class of valuation functions for which a constant factor approximation guarantee with a polynomial time (with a value oracle), truthful and budget feasible mechanism is known. For subadditive valuations 
\cite{dobzinski2011} introduced a 
mechanism using a demand oracle (more powerful than the value oracle we use). The current best bound is an $O(\frac{\log n}{\log \log n})$ approximation guarantee due to \cite{bei2012}. \cite{bei2012} also gave randomized mechanism that achieves a constant (768) approximation guarantee for fractionally subadditive (XOS) valuations, also using a demand oracle.

Some papers consider the Bayesian setting, where cost of each agent comes from known independent distributions. \cite{bei2012} gave a constant-competitive  mechanism for subadditive valuations (with a very large constant).
\cite{jason2016} gave a $(\frac{e}{e-1})^2$-competitive posted pricing mechanism for monotone submodular valuations for large markets, using a cost version for defining the largeness of the market. The benchmark (optimum) used in \cite{jason2016} is the outcome of optimal Bayesian incentive compatible mechanism, while others (including us) have used the significantly higher, optimum  with respect to the budgeted pure optimization problem as their benchmark. It is interesting to compare our results for large markets to the approximation guarantee of $(\frac{e}{e-1})^2 \approx 2.5$ of the mechanism in \cite{jason2016}. While this bound is $\approx 0.08$ better than our bound, their benchmark, the optimal Bayesian incentive compatible mechanism, can be a factor of $\frac{e}{e-1}$ lower than our benchmark of the optimum ignoring incentives \cite{anari2014}. Even when the cost of sellers come from a uniform distributions, and the value of each item is 1, the ratio between the two benchmarks is $\sqrt{2}$.

\section{Preliminaries}
\label{Sec:Prelims}
We consider the problem of a single buyer with a budget $B$ facing a set of multiple sellers $A$, each selling a single item. We let $n$ denote the number of sellers and we assume $A=[n]$. We assume that the value $v(S)$ of the buyer for a set of items $S$, is a  (non-decreasing) submodular function, that is, satisfies $v(S) \leq v(T)$ for every $S \subseteq T$, and  $v(S) + v(T) \geq v(S \cup T) + v(S \cap T)$, for every set $S,T \subseteq A$. For every $i\in A$ and $S \subseteq A$, we define $m_i(S)=v(S \cup \{i\})-v(S)$, i.e. the marginal value of $i$ with respect to subset $S$. Note that $v(.)$ is monotone submodular if and only if for every $S,T \subseteq A$ we have:
$$v(T) \leq v(S) + \sum_{i \in T \setminus S}m_i(S).$$

\textit{The large market assumption. } In Section \ref{Sec:LargeMarkets}, we consider large markets, assuming that the value of each agent is small compared to the optimum, i.e. $v(i) \ll opt(A)$ for all $i \in [n]$. For simplicity, we state our approximation bounds for large markets in the limit\footnote{By having a $\theta$-large market assumption instead, the approximation guarantees for our large market mechanisms increases by a factor of $(1-c\theta)^{-1}$, where $c \in (0,4)$ is a constant which is different for each mechanism. We omit stating the exact value of $c$ for each mechanism separately.}, assuming $\theta = \max_{i \in [n]}\frac{v(i)}{opt(A)}\rightarrow 0$.

The mechanism design problem of selecting sellers maximizing the buyer's value subject to his budget constraint, is a single parameter mechanism design problem, in which each bidder (seller) has one private value (the cost of her item). We design truthful deterministic and individually rational mechanisms, as well as universally truthful and individually rational randomized mechanisms. A randomized mechanism is \emph{universally truthful} if it is a randomization among deterministic mechanisms, each of which are truthful. We use Myerson's characterization for truthful mechanisms (theorem \ref{THM:Mayerson}), stating that a mechanism is truthful and individually rational if and only if the choice of selecting each item is monotone in its declared cost, and winners are paid threshold payments that are above their declared cost.

\begin{theorem}\label{THM:Mayerson}\cite{myerson1981}
In single parameter domains, a normalized mechanism $M=(f,p)$ is truthful if and only if
\begin{itemize}
\item \textbf{$f$ is monotone}: $\forall i \in [n]$, if $c'_i \leq c_i$ then $i \in f(c_i,c_{-i})$ implies $i \in f(c'_i,c_{-i})$, or equivalently, $c'_i \notin f(c'_i,c_{-i})$ implies $c_i \notin f(c_i,c_{-i})$.
\item \textbf{Winners are paid threshold payments}: if $i \in [A]$ is a winner and receives payment $p_i$, then $p_i = \inf \{c_i:i\notin f(c_i,c_{-i})\}$.
\end{itemize}
\end{theorem}

In order to show a mechanism is universally truthful and budget feasible, it suffices to show that the allocation is monotone and by using the threshold payments, the total payments are not more than the budget. Similar to \cite{dobzinski2011,chen2011,bei2012,singla2013}, we assume that the payments are threshold payments and only specify the allocation rule. At the end of each section, we briefly explain how the payment rule of the mechanisms in that section can be computed. \pjedit{In all our mechanisms if a seller bids a cost more than $B$, he will not be selected in the winning set, hence will have utility 0. This combined with the fact that all our mechanisms are truthful, implies \textit{individual rationality}, i.e. in all of our mechanisms utility of sellers are non-negative.}

\section{Parameterized Mechanisms for Submodular Valuations}
\label{Sec:Param}
In this section we present two simple parameterized mechanisms. We show that these parameterized mechanisms  provide good approximation guarantees, and are monotone and hence can be turned into truthful mechanisms with payments defined appropriately. We analyze the approximation guarantee of these mechanisms with and without the large market assumption and give conditions that make these mechanisms budget feasible.

Let $S_0=\emptyset$, and for each $i \in [n]$, recursively define $S_i=S_{i-1} \cup \{\argmax_{j \in A \setminus S_{i-1}}( \frac{m_{j}(S_{i-1})}{c_j})\}$, adding the item with maximum marginal value to cost ratio, to $S_{i-1}$. To simplify notation, we will assume without loss of generality that $\{i\}=S_i\setminus S_{i-1}$. All of our mechanisms sort the items in descending order of marginal bang for buck at the beginning \etedit{and consider items in this order}.

\subsection{The Threshold Mechanism}
Our threshold mechanism is a framework generalizing the mechanisms of
Singer \cite{singer2010} and Chen et al \cite{chen2011}.
We consider items in increasing cost-to-marginal value order, as defined above.
Our \emph{greedy threshold mechanism}, \textsc{Greedy-TM}, sets a threshold for the cost to marginal value ratio of the items, compared to the ratio of the budget to the total value of the set we have so far.  Using a parameter $\gamma$, the mechanism adds items while they are relatively cheap compared to the total so far.

The \textsc{Greedy-TM} mechanism works well for large markets where each individual item has small value compared to the optimum. In the general case, we will randomly choose between just selecting the item with maximum individual value and cost below the budget, or running \textsc{Greedy-TM}. We call the resulting randomized mechanism \textsc{Random-TM}$(\gamma, A, B)$.

\begin{minipage}[t]{7.5cm}
\IncMargin{2em}
\begin{algorithm}[H]
\label{Mech:Gamma}
\DontPrintSemicolon
\textsc{Greedy-TM}$(\gamma,A,B)$\;
\indent (Greedy Threshold Mechanism)\;
Let $k=1$\;
\While{$k\leq |A|$ and $\frac{c_{k}}{m_{k}(S_{k-1})} \leq \gamma \frac{B}{v(S_{k})}$}{
	$k=k+1$
	}
\Return $S_{k-1}$
\end{algorithm}
\IncMargin{-2em}
\end{minipage}
\begin{minipage}[t]{6.5cm}
\IncMargin{2em}
\begin{algorithm}[H]
\label{Mech:GammaGeneral}
\DontPrintSemicolon
\textsc{Random-TM}$(\gamma, A, B)$\;
\indent (Random Threshold Mechanism)\;
Let $A= \{i:c_i\leq B\}$\;
Let $i^*=argmax_{i\in[n]}(v(i))$\;
\textbf{With probability $\frac{\gamma+1}{\gamma+2}$ do}\;
\Indp{
 \Return \textsc{Greedy-TM}$(\gamma,A,B)$\;
 \textbf{halt}\;
 }
\Indm
\Return $i^*$
\end{algorithm}
\IncMargin{-2em}
\end{minipage}

The randomized mechanisms for submodular functions in Singer \cite{singer2010} is similar to \textsc{Random-TM} with parameter $\gamma=\frac{e-1}{12e-4}$ and the improved mechanism of Chen et al \cite{chen2011} is equivalent to \textsc{Random-TM} with $\gamma=0.5$.

Monotonicity of the mechanisms is easy to see: if someone is not chosen, 
he cannot be selected by increasing his cost (decreasing his marginal bang per buck).
\begin{lemma}\label{Lem:GammaMon}
For every fixed $\gamma \in (0,1]$, the mechanism \textsc{Greedy-TM}$(\gamma,A,B)$ is monotone.
\end{lemma}

We show that for every fixed $\gamma \in (0,1]$, \textsc{Random-TM}$(\gamma, A, B)$  achieves a $1+\frac{2}{\gamma}$ approximation of the optimum, improving the bound of \cite{chen2011}. The key difference is that we compare the output of \textsc{Greedy-TM}
directly with the true optimum, rather than a fractional greedy solution.
 Doing this not only improves the approximation factor, but also simplifies the analysis.

\begin{lemma}\label{Lem:GammaApprox}
For every fixed $\gamma \in (0,1]$, if $S_{k-1}=$ \textsc{Greedy-TM}$(\gamma,A,B)$ then
$$(1+\frac{1}{\gamma})v(S_{k-1}) + \frac{1}{\gamma}v(i^*)\geq opt(A)$$
\end{lemma}
\begin{proof}
Let $S^*$ be the optimum. By monotonicity and submodularity of $v(.)$, we have
\begin{align*}
v(S^*)-v(S_{k-1}) &\leq v(S^* \cup S_{k-1}) - v(S_{k-1})
\leq \sum_{i \in S^*\setminus S_{k-1}}m_{i}(S_{k-1})=\sum_{i \in S^*\setminus S_{k-1}}c_i\frac{m_{i}(S_{k-1})}{c_i}
\end{align*}
By using the fact that $k \in \argmax_{i \in A\setminus S_{k-1}}\frac{m_{i}(S_{k-1})}{c_i}$, we get $\frac{m_{i}(S_{k-1})}{c_i} \le \frac{m_{k}(S_{k-1})}{c_k}$.
%
Since $k$ is not in the winning set $\frac{c_k}{m_{k}(S_{k-1})}>\gamma \frac{B}{v(S_k)}$. \etedit{Using these} we get
\begin{align*}
\sum_{i \in S^*\setminus S_{k-1}}c_i\frac{m_{i}(S_{k-1})}{c_i} &\le c(S^* \setminus S_{k-1}) \frac{m_{k}(S_{k-1})}{c_{k}}< B (\frac{1}{\gamma}\frac{v(S_{k})}{B}) =\frac{1}{\gamma} v(S_{k})
\end{align*}
Finally, by definition of $v(S_k)$ and using submodularity, we have
\begin{align*}
\frac{1}{\gamma} v(S_{k})\leq \frac{1}{\gamma} (v(S_{k-1})+v(k))
\leq \frac{1}{\gamma} (v(S_{k-1})+v(i^*))
\end{align*}
By putting all the above together and rearranging the terms, we get the desired inequality.
\end{proof}

By using the above lemma for the performance of \textsc{Random-TM}, we can get the following approximation bound for \textsc{Random-TM}$(\gamma, A, B)$.

\begin{theorem}\label{THM:GammaApprox}
For every fixed $\gamma \in (0,1]$, \textsc{Random-TM}$(\gamma, A, B)$ is universally truthful, and  has approximation ratio of $1+\frac{2}{\gamma}$.
\end{theorem}
\begin{proof}
By Lemma \ref{Lem:GammaMon} the mechanism is monotone and hence universally truthful.

To prove the approximation ratio, let $S$ be the outcome of the mechanism. The mechanism chooses $S_{k-1}=$ \textsc{Greedy-TM}$(\gamma,A,B)$ with probability $\frac{\gamma+1}{\gamma+2}$ and $i^*$ with probability $1- \frac{\gamma+1}{\gamma+2}= \frac{1}{\gamma+2}$. So we have
\begin{align*}
&E[v(S)] = \frac{\gamma+1}{\gamma+2}v(S_{k-1}) + \frac{1}{\gamma+2}v(i^*)\\
\Rightarrow & (1+\frac{2}{\gamma})E[v(S)] = \frac{\gamma+1}{\gamma}v(S_{k-1}) + \frac{1}{\gamma}v(i^*)
\end{align*}
So by using lemma \ref{Lem:GammaApprox} we have
$$(1+\frac{2}{\gamma})E[v(S)]\geq opt(A)$$
\end{proof}

The mechanisms \textsc{Greedy-TM}$(\gamma,A,B)$ and \textsc{Random-TM}$(\gamma, A, B)$ are not necessarily budget feasible for an arbitrary choice of $\gamma$. However, \cite{chen2011} shows that \textsc{Random-TM}$(0.5, A, B)$ (which they call \textsc{Random-SM}) is budget feasible. We include a simplified proof in Section \ref{Sec:Appendix} for completeness.

\begin{theorem}[\cite{chen2011}]\label{THM:Chen}
\textsc{Random-TM}$(0.5,A,B)$ and \textsc{Greedy-TM}$(0.5,A,B)$ are budget feasible.
\end{theorem}

Combining Theorem \ref{THM:Chen} and Theorem \ref{THM:GammaApprox} for the general case, and  using Lemma \ref{Lem:GammaApprox} directly, instead of Theorem \ref{THM:GammaApprox} for the case of large market, where $v(i^*)\ll opt(A)$, we get the following theorem. The bound for large markets is
matching the best approximation guarantee of Anari et al \cite{anari2014} for submodular functions with computational constraint. In section \ref{Sec:LargeMarkets} we improve this bound.

\begin{corollary}
\textsc{Random-TM}$(0.5, A, B)$ is truthful, budget feasible and has approximation ratio of $5$. For the case of  large market case, where $v(i^*)\ll opt(A)$, \textsc{Greedy-TM}$(\gamma, A, B)$ is is truthful, budget feasible and has approximation ratio of $3$.
\end{corollary}

We note that our approximation analysis for \textsc{Random-TM}$(0.5, A, B)$ is tight. To see this consider the following example (with additive valuation): assume we have 5 items numbered from 1 to 5 with budget 4. Let $v_1=1$ and $c_1=0$, and let $v_i=1-\epsilon$ and $c_i=1$ for $2 \leq i \leq 5$. The mechanism chooses item 1, however since $1> \frac{4}{2} \frac{1-\epsilon}{2-\epsilon}$, none of the other items is in the winning set of \textsc{Greedy-TM}$(0.5,A,4)$. So the value of the \textsc{Greedy-TM}$(0.5,A,4)$ as well as the value of \textsc{Random-TM}$(0.5,A,4)$ is 1. However, optimum can select all of the items and get the value $5-4\epsilon$. Since $\epsilon$ can be arbitrarily small, \textsc{Random-TM}$(0.5,A,B)$ is at most a 5-approximation.

The threshold payment of each agent $i$ in the winning set for the threshold mechanisms in this section can be computed by increasing $i$'s cost until he reaches the threshold that makes him not eligible to be in the winning set, while keeping the cost of other agents fixed. In order to compute this number in polynomial time, it is enough to fix other agents' costs and see where in the sorted list of marginal bang-per-bucks this agent can be appear such that the inequality of \textsc{Greedy-TM}$(\gamma,A,B)$ still holds for her. The more detailed characterization of these threshold payments is similar to that of \cite{singer2010}.



\label{Sec:Param1}
\subsection{The Oracle Mechanism}
Here, we provide a different class of parameterized mechanisms. This class of mechanisms requires an oracle $Oracle(A,B)$, which considers the optimization problem of maximizing the value subject to the total cost not exceeding the budget $B$, and returns a value which is close to optimum. Let $opt(A,B)$ denote the optimum value of this optimization problem. We assume that $opt(A,B) \geq Oracle(A,B)$. The oracle is an $r$ approximation, if we also have $r \cdot Oracle(A,B) \geq opt(A,B)$. For instance the greedy algorithm of Sviridenko \cite{greedy2004} can be used as an oracle with  $r=\frac{e}{e-1}\approx 1.58$.

\paragraph{Exponential time mechanism.}
We start with a simple \emph{exponential time oracle mechanism}, \textsc{Greedy-EOM}, using the optimal solution value $opt(A,B)$ as an oracle. The optimum value is simpler to use, as it is monotone in the cost of the agents\footnote{For large markets, Anari et al \cite{anari2014} use an exponential time mechanism which is similar to \textsc{Greedy-EOM}$(0.5,A,B)$}. Later, we show how to use a polynomial time approximation oracle instead of $opt(A,B)$, with a small sacrifice in the approximation ratio while keeping the mechanism truthful and budget feasible. Our mechanisms in this section also sort the items in decreasing order of bang-per-buck, as explained in the beginning of the section.

\hspace{-0.3in}
\begin{minipage}[t]{4.8cm}
\IncMargin{2em}
\begin{algorithm}[H]
\label{Mech:EAlphaGreedy}
\DontPrintSemicolon
\textsc{Greedy-EOM}$(\alpha,A,B)$\;
(exp. time oracle mechanism)\;
Let $v^* = opt(A,B)$\;
Let $k=1$\;
\While{$v(S_k) \leq \alpha v^*$}{
	$k=k+1$
	}
\Return $S_{k-1}$
\end{algorithm}
\IncMargin{-2em}
\end{minipage}
\begin{minipage}[t]{5.5cm}
\IncMargin{2em}
\begin{algorithm}[H]
\label{Mech:EAlphaGeneral}
\DontPrintSemicolon
\textsc{Random-EOM}$(\alpha, A, B)$\;
(exp. time oracle mechanism)\;
Let $A= \{i:c_i\leq B\}$\;
Let $i^*=argmax_{i\in[n]}(v(i))$\;
\textbf{With probability $\frac{1}{2}$ do}\;
\Indp{
 \Return \textsc{Greedy-EOM}$(\alpha,A,B)$\;
 \textbf{halt}\;
 }
\Indm
\Return $i^*$
\end{algorithm}
\IncMargin{-2em}
\end{minipage}
\begin{minipage}[t]{5.7cm}
\IncMargin{2em}
\begin{algorithm}[H]
\label{EAlphaGeneral}
\DontPrintSemicolon
\textsc{Deterministic-EOM}$(\alpha, A, B)$\;
(det. exp. time oracle mechanism)\;
Let $A= \{i:c_i\leq B\}$\;
Let $i^*=argmax_{i\in[n]}(v(i))$\;
Let $v=opt(A\setminus \{i^*\},B)$\;
\uIf{$v(i^*)\geq \frac{\sqrt{17}-3}{4} v$}{
\Return $i^*$
}\Else{
\Return \textsc{Greedy-EOM}$(0.5,A,B)$\;
}
\end{algorithm}
\IncMargin{-2em}
\end{minipage}

\begin{lemma}\label{Lem:EAlphaApprox}
For every fixed $\alpha \in (0,1]$, \textsc{Greedy-EOM}$(\alpha,A,B)$ is monotone, and if $S_{k-1}=$ \textsc{Greedy-EOM}$(\alpha,A,B)$ then $\frac{1}{\alpha}v(S_{k-1}) + \frac{1}{\alpha}v(i^*)\geq opt(A,B)$.
\end{lemma}
\begin{proof}
We first argue that the mechanism is monotone. Assume $i \in A$ and $i \notin S_{k-1}$. If $i$ increase his cost, it cannot increase the value of $v^*=opt(A,B)$. Furthermore, 
by increasing $i$'s cost, her marginal bang per buck decreases, which cannot help him get selected, 
so the mechanism is monotone.

To see the approximation bound simply note that by the definition of the mechanism we have $v(S_{k})>\alpha v^*$. So $v(S_{k-1})+v(i^*) > v(S_{k-1})+v(k) > \alpha v^*$. 
\end{proof}

By using lemma \ref{Lem:EAlphaApprox}, it is easy to prove the following theorem.
\begin{theorem}\label{THM:EAlphaApprox}
For every fixed $\alpha \in (0,1]$, \textsc{Random-EOM}$(\alpha, A, B)$ is universally truthful, and if $S=$ \textsc{Random-EOM}$(\alpha, A, B)$ then $\frac{2}{\alpha}E[v(S)]\geq opt(A)$.
\end{theorem}
\begin{proof}
By using lemma \ref{Lem:EAlphaApprox} and similar argument to proof of truthfulness in theorem \ref{THM:GammaApprox}, it is easy to see that the mechanism is universally truthful.

By definition of \textsc{Random-EOM}$(\alpha, A, B)$ we have
\begin{align*}
&E[v(S)] = \frac{1}{2}v(S) + \frac{1}{2}v(i^*)\\
\Rightarrow &\frac{2}{\alpha} E[v(s)] = \frac{1}{\alpha}v(S) + \frac{1}{\alpha} v(i^*)
\end{align*}
By using lemma \ref{Lem:EAlphaApprox} the proof is complete.
\end{proof}

Now we show that for the choice of $\alpha = 0.5$, \textsc{Greedy-EOM}$(0.5,A,B)$ is budget feasible, so \textsc{Random-EOM}$(0.5, A, B)$ is universally truthful, budget feasible, and a 4 approximation to the optimum. \etnedit{A bit more complex analog of this lemma for the mechanism using an approximation algorithm in place of the true optimum will be lemma \ref{Lem:AlphaBudget}.}

\begin{lemma}\label{Lem:EAlphaBudget}
By using threshold payments, \textsc{Greedy-EOM}$(0.5,A,B)$ is budget feasible.
\end{lemma}
\begin{proof}
Let $p_i$ be the threshold payment for agent $i$. Let $S_{k-1}=$\textsc{Greedy-EOM}$(0.5,A,B)$. For every $i \in S_{k-1}$, we show that if $i$ deviates to a bid of $b_{i} > m_{i}(S_{i-1})\frac{B}{v(S_{k-1})}$, he cannot be selected, implying that the threshold payment $p_i \le m_{i}(S_{i-1})\frac{B}{v(S_{k-1})}$.
By proving this we get that
$\sum_{i \in S_{k-1}}p_{i} \leq \sum_{i \in S_{k-1}}m_{i}(S_{i-1}) \frac{B}{v(S_{k-1})}=B$, where the last inequality hold by recalling that $m_{i}(S_{i-1})=v(S_{i})-v(S_{i-1})$, so the sum telescopes.

So the mechanism is budget feasible.

We prove the inequality claimed above by contradiction: assume that $i$ deviates to $b_{i}>m_{i}(S_{i-1})\frac{B}{v(S_{k-1})}$ and is still in the winning set. Let $b$ be the new vector of costs with $i$ bidding $b_i$ and all other agents bidding their true cost. Note that the order in which items are considered after item $i-1$ is also effected by the change in $i$'s claimed cost. Now let $j$ be the step in the mechanism in which $i$ is added to the winning set after he deviates to $b_{i}$ and $S'_j$ be the wining set after that step, where $S'_z$ for $z\in [n]$ is defined similar to $S_z$ but with cost vector $b$ instead of $c$. Let $S^*$ be the optimum solution with $v(S^*)=v^*$. Let $S^*\setminus S'_j=\{t_1,t_2,\ldots,t_q\}$, $T_0=\emptyset$, and $T_z=\{t_l:l\in[z]\}$. Since $i$ is the only item that has increased his cost and $i \in S'_{j}$, we have
\begin{align*}
c(S^*)&\geq b(S^*\cup S'_{j})-b(S'_{j})= \sum_{z \in [q]} m_{t_z}(S'_{j}\cup T_{z-1})\frac{b_{t_z}}{m_{t_z}(S'_{j}\cup T_{z-1})}
\end{align*}
By submodularity and by the fact that the mechanism choose the ordering having item $i$ in position $j$ (with costs $b$) we have that $m_{t_z}(S'_{j}\cup T_{z-1})\le m_{t_z}(S'_{j-1})$ and $b_{t_z}/m_{t_z}(S'_{j-1})\le b_i/m_{i}(S'_{j-1})$. Using these two inequalities we get
\begin{align*}
\sum_{z \in [q]} m_{t_z}(S'_{j}\cup T_{z-1})\frac{b_{t_z}}{m_{t_z}(S'_{j}\cup T_{z-1})} & \geq \sum_{z \in [q]} m_{t_z}(S'_{j}\cup T_{z-1})\frac{b_{t_z}}{m_{t_z}(S'_{j-1})} \geq \sum_{z \in [q]} m_{t_z}(S'_{j}\cup T_{z-1})\frac{b_{i}}{m_{i}(S'_{j-1})}\\
& 
= \frac{b_{i}}{m_{i}(S'_{j-1})} (v(S^*\cup S'_{j})-v(S'_{j}))
\geq \frac{b_{i}}{m_{i}(S'_{j-1})} (v(S^*)-v(S'_{j}))
\end{align*}
Now by using the contradiction assumption, and the fact that $S_{i-1} \subseteq S'_{j-1}$, we get
\begin{align*}
\frac{b_{i}}{m_{i}(S'_{j-1})} (v(S^*)-v(S'_{j})) > B\frac{v(S^*)-v(S'_{j})}{v(S_{k-1})}
\end{align*}
By combining the above inequalities and using the fact that $v(S'_{j}),v(S_{k-1})<\alpha v^*$ and $\alpha =0.5$, we have $c(S^*)>B$ which is a contradiction, so the mechanism is budget feasible.
\end{proof}

\begin{corollary}
By using threshold payments, \textsc{Random-EOM}$(0.5, A, B)$ is universally truthful, budget feasible and a 4 approximation of the optimum.
\end{corollary}
Next we offer a simple deterministic version of this mechanism, with a significantly better approximation factor\etnedit{, improving the previously known 8.34 approximation exponential time mechanism
of \cite{chen2011} to a guarantee of 4.56}. In order to derandomize \textsc{Random-EOM}, we would like to check if the optimum is large enough compared to the best valued item. To keep the mechanisms monotone, we will compare the value of the highest valued item $i^*$ to the optimum after removing item $i^*$.

\begin{theorem}\label{THM:EAlphaDet}
By using threshold payments, \textsc{Deterministic-EOM} is truthful, budget feasible and has an approximation ratio of 4.56.
\end{theorem}
\begin{proof}
$i^*$ cannot change $opt(A \setminus \{i^*\},B)$, so since \textsc{Greedy-EOM}$(0.5,A,B)$ is monotone and budget feasible, \textsc{Deterministic-EOM} is truthful and budget feasible.

In order to prove the approximation ratio, we consider two cases:
\begin{itemize}
\item[Case 1:] $v(i^*)\geq \frac{\sqrt{17}-3}{4} v=\etnedit{\frac{\sqrt{17}-3}{4} }opt(A\setminus \{ i^* \},B)$: in this case the algorithm returns $i^*$, and we have
$$(\frac{4}{\sqrt{17}-3} +1)v(i^*)\geq opt(A\setminus \{i^*\},B) +v(i^*) \geq opt(A,B)$$

\item[Case 2:] $v(i^*)< \frac{\sqrt{17}-3}{4} v \leq \frac{\sqrt{17}-3}{4} opt(A,B)$: let $S_{k-1}=$ \textsc{Greedy-EOM}$(0.5,A,B)$. In this case, by lemma \ref{Lem:EAlphaApprox}, we have
\begin{align*}
&2v(S_{k-1})+2v(i^*)\geq opt(A)\\
\Rightarrow &2v(S_{k-1})+2(\frac{\sqrt{17}-3}{4})opt(A)\geq opt(A)\\
\Rightarrow &\frac{2}{1-\frac{\sqrt{17}-3}{2}}v(S_{k-1})\geq opt(A)
\end{align*}
\end{itemize}
So in any case the approximation ratio is $\frac{4}{\sqrt{17}-3} +1 = \frac{2}{1-\frac{\sqrt{17}-3}{2}}\simeq 4.56$
\end{proof}

\paragraph{The Polynomial Time Oracle Mechanism.}
Next we offer a version of the above mechanism using an oracle in place of the optimum value $opt(A,B)$, as finding the optimum for monotone submodular maximization with a knapsack constraint cannot be done in polynomial time. Note that naively using the outcome of an oracle, which is not optimum, can break monotonicity. To see this, note that in proof of monotonicity for \textsc{Greeding-EOM}$(\alpha,A,B)$, we use the fact that if an item increases his cost, he cannot increase the value of $opt(A,B)$. If we replace $opt(A,B)$ with the outcome of a sub-optimal oracle (for instance a greedy algorithm), this is no longer true: if one increases the cost of all the items that are not in the optimum set to be more than the budget, any reasonable approximation algorithm for submodular maximization (for instance the greedy algorithm in \cite{greedy2004}) can detect and choose all the items that are in optimum set. 

To make our mechanism monotone, we remove $i$ before calling the oracle to decide if we should add $i$ to the set $S$, 
making the items selected  no longer contiguous in the order we consider them.

\begin{minipage}[t]{7cm}
\IncMargin{2em}
\begin{algorithm}[H]
\label{Mech:AlphaGreedy}
\DontPrintSemicolon
\textsc{Greedy-OM}$(\alpha,A,B)$\;
\indent (Greedy Oracle Mechanism)\;
Let $S=\emptyset$\;
\For{$i=1$ to $n$}{
	\If{$v(S_i)\leq \alpha Oracle(A\setminus \{i\},B)$}
	{
		$S=S\cup\{i\}$
	}
	}
\Return $S$
\end{algorithm}
\IncMargin{-2em}
\end{minipage}
\begin{minipage}[t]{7cm}
\IncMargin{2em}
\begin{algorithm}[H]
\label{Mech:AlphaGeneral}
\DontPrintSemicolon
\textsc{Random-OM}$(\alpha,A,B)$\;
\indent (Random Oracle Mechanism)\;
Let $A= \{i:c_i\leq B\}$\;
Let $i^*=argmax_{i\in[n]}(v(i))$\;
\textbf{With probability $\frac{r}{\alpha+2r}$ do}\;
\Indp{
 \Return \textsc{Greedy-OM}$(\alpha,A,B)$\;
 \textbf{halt}\;
 }
\Indm
\Return $i^*$
\end{algorithm}
\IncMargin{-2em}
\end{minipage}

Next we show that \textsc{Greedy-OM}$(\alpha,A,B)$ is monotone and provide its approximation ratio.

\begin{lemma}\label{Lem:AlphaAprox}
For every fixed $\alpha \in (0,1]$, \textsc{Greedy-OM}$(\alpha,A,B)$ is monotone. If $S=$ \textsc{Greedy-EOM}$(\alpha,A,B)$, $k\in[n]$ is the biggest integer such that $S_{k-1}\subseteq S$, $i^*$ is the item with maximum individual value, and assuming \textsc{Oracle} is an $r$ approximation of the optimum, then
$$\frac{r}{\alpha}v(S_{k-1}) + (1+\frac{r}{\alpha})v(i^*)\geq opt(A)$$
\end{lemma}
\begin{proof}
Monotonicity of the mechanism follows from the usual argument, increasing $c_i$ does not effect \textsc{Oracle}$(A \setminus \{i\},B)$ and decreases the item's bang per buck in any step. To show the approximation factor, recall that $\{k\} = S_{k} \setminus S_{k-1}$. Since $k$ was not chosen by the mechanism we have
\begin{align*}
v(S_{k-1})+ v(k)&\geq v(S_k) > \alpha Oracle(A \setminus \{k\})
\geq \frac{\alpha}{r}opt(A \setminus \{k\})\\
& \geq \frac{\alpha}{r}(opt(A)-v(k))\geq \frac{\alpha}{r}(opt(A)-v(i^*))
\end{align*}
\end{proof}

Next we show that \textsc{Greedy-OM}$(0.5,A,B)$ is budget feasible.

\begin{lemma}\label{Lem:AlphaBudget}
By using threshold payments, \textsc{Greedy-OM}$(0.5,A,B)$ is budget feasible
\end{lemma}
\begin{proof}
Let $p_i$ be the threshold payment for agent $i$. Let $S=$\textsc{Greedy-OM}$(0.5,A,B)$. For every $i \in S$, we show that if $i$ deviates to a cost $b_{i} > m_{i}(S_{i-1})\frac{B}{v(S)}$, he cannot be selected. By proving this and by using the definition of threshold payments we get
$\sum_{i \in S}p_{i} \leq \sum_{i \in S}m_{i}(S_{i-1}) B/v(S) \leq \sum_{i \in S}m_{i}(S_{i-1}\cap S) B/v(S) = B$, so the mechanism is budget feasible.

Assume that $i$ deviates to $b_{i}>m_{i}(S_{i-1})\frac{B}{v(S)}$ and is still in the winning set. Let $b$ be the resulting costs, with $i$'s cost as $b_i$ and for all other players $b_j=c_j$. Note that this change in cost for $i$ changes the order in which the mechanism considers the items.
Let $j$ be the step that the mechanism considers item $i$ with costs $b$. For $z\in [n]$, $S'_z$ is defined similar to $S_z$ but with cost vector $b$ instead of $c$.  \footnote{Note that in \textsc{Greedy-TM}, as well as \textsc{Greedy-EOM}, $S'_j$ would be exactly the winning set of the mechanism at the end of step $j$. However, in \textsc{Greedy-OM}, $S'_j$ may be a super set of the set of items that has been added to the winning set at the end of step $j$.} 
\begin{align*}
c(S^*)&\geq b(S^*\cup S'_{j})-b(S'_{j})= \sum_{z \in [q]} m_{t_z}(S'_{j}\cup T_{z-1})\frac{b_{t_z}}{m_{t_z}(S'_{j}\cup T_{z-1})}
\end{align*}
By submodularity we have  $m_{t_z}(S'_{j})\ge m_{t_z}(S'_{j}\cup T_{z-1})$, and by the fact that $i$ was in position $j$ in the ordering considered,  we have that $b_{t_z}/m_{t_z}(S'_{j-1})\ge b_{i}/m_{i}(S'_{j-1})$. Using these we get:
\begin{align*}
\sum_{z \in [q]} m_{t_z}(S'_{j}\cup T_{z-1})\frac{b_{t_z}}{m_{t_z}(S'_{j}\cup T_{z-1})} &\geq \sum_{z \in [q]} m_{t_z}(S'_{j}\cup T_{z-1})\frac{b_{t_z}}{m_{t_z}(S'_{j-1})}
\geq \sum_{z \in [q]} m_{t_z}(S'_{j}\cup T_{z-1})\frac{b_{i}}{m_{i}(S'_{j-1})}\\
&= \frac{b_{i}}{m_{i}(S'_{j-1})} (v(S^*\cup S'_{j})-v(S'_{j}))
\geq \frac{b_{i}}{m_{i}(S'_{j-1})} (v(S^*)-v(S'_{j}))
\end{align*}
Now by using fact that $S_{i-1} \subseteq S'_{j-1}$, as the first $i-1$ steps of the algorithm are not effected by $i$'s change of bid, and by the assumption about $b_i$, we have
\begin{align*}
\frac{b_{i}}{m_{i}(S'_{j-1})} (v(S^*)-v(S'_{j})) > B\frac{v(S^*)-v(S'_{j})}{v(S)}
\end{align*}
So by combining the above inequalities we have that the total cost $c(S^*)$ of set $S^*$ is $c(S^*) >B (v(S^*)-v(S'_j))/v(S)$ .

Since $i$ was selected after deviating, we have $v(S'_j) \leq \alpha \textsc{Oracle}(A \setminus \{a_i\},B) \leq \alpha v(S^*)$. Let $k'\in [n]$ be the minimum integer such that $S \subseteq S_{k'}$. Since $k'$ was chosen, we have $v(S)\leq v(S_{k'}) \leq \alpha Oracle(A \setminus \{a_{k'}\},B) \leq \alpha v(S^*)$. Since $\alpha=0.5$ we get $C(S^*)>B$ which is contradiction.
\end{proof}

Note that in large markets $v(i^*)\ll opt(A)$, so by combining lemma \ref{Lem:AlphaAprox} and lemma \ref{Lem:AlphaBudget}, we get the following corollary.

\begin{corollary}\label{Cor:AlphaBudget}
In large markets, \textsc{Greedy-OM}$(0.5,A,B)$ is truthful, budget feasible and given an $r$-approximation oracle, achieves $2r$ approximation of the optimum.
\end{corollary}
The previously known best oracle mechanism for large markets is due to Anari et al \cite{anari2014} and achieves $2r^2$. We will further improve this bound for the case of large markets \etnedit{to $r+1$} in section \ref{Sec:LargeMarkets}.

By using lemma \ref{Lem:AlphaAprox}, we get the following theorem.
\begin{theorem}\label{THM:AlphaApprox}
\textsc{Random-OM}$(\alpha,A,B)$ is truthful and in expectation achieves $1+\frac{2r}{\alpha}$ of the optimum, assuming the oracle used is an $r$-approximation.
\end{theorem}
\begin{proof}
Let $S$ be the outcome of the mechanism. The mechanism chooses $S=$ \textsc{Greedy-OM}$(\gamma,A,B)$ with probability $\frac{r}{\alpha+2r}$ and $i^*$ with probability $1- \frac{r}{\alpha+2r}= \frac{\alpha +r}{\alpha+2r}$. Let $k\in [n]$ be the biggest integer such that $S_{k-1} \subseteq S$. So we have
\begin{align*}
&E[v(S)] \geq \frac{r}{\alpha+2r}v(S_{k-1}) + \frac{\alpha+r}{\alpha+2r}v(i^*)\\
\Rightarrow & (1+\frac{2r}{\alpha})E[v(S)] \etnedit{\ge} \frac{r}{\alpha}v(S_{k-1}) + \frac{\alpha+r}{\alpha}v(i^*)
\end{align*}
So by using lemma \ref{Lem:AlphaAprox} we have
$$(1+\frac{2r}{\alpha})E[v(S)]\geq opt(A)$$
\end{proof}

By combining lemma \ref{Lem:AlphaBudget} and theorem \ref{THM:AlphaApprox} we have the following theorem.

\begin{theorem}\label{THM:Alpha}
\textsc{Random-OM}$(0.5,A,B)$ is truthful, budget feasible and in expectation achieves $1+4r$ of the optimum.
\end{theorem}


By using the greedy algorithm of \cite{greedy2004} as an oracle, we can improve the approximation ratio to $\frac{2r}{\alpha}$. To achieve this, we change \textsc{Greedy-OM}$(\alpha,A,B)$, so that instead of using $Oracle(A\setminus \{k\},B)$, it uses $\max_{c'_{i}\geq c_{i}}Oracle(A,(c'_{i},c_{-i}))$. We also change the probability of choosing the greedy mechanism's outcome in \textsc{Random-OM}$(\alpha,A,B)$ to $\frac{1}{2}$. By doing so \textsc{Random-OM}$(\alpha,A,B)$ can achieve $\frac{2r}{\alpha}$ instead of $1+\frac{2r}{\alpha}$. By using the greedy algorithm of \cite{greedy2004}, as an oracle, finding $\max_{c_{i}'\geq c_{i}}Oracle(A,(c_{i}',c_{-i}))$ can be done in polynomial time, since we only have to check polynomial number of cases for $c'_{i}$. Furthermore, if $i$ increases his cost, he cannot increase the value of $\max_{c_{i}'\geq c_{i}}Oracle(A,(c_{i}',c_{-i}))$. We omit the proof of the following theorem, as it is analogous to our previous proofs.
\begin{theorem}\label{THM:Alpha-greedy}
The above modification of the \textsc{Random-OM}$(0.5,A,B)$ mechanism is truthful, budget feasible, get expected value of a $4r$ fraction of the optimum. With the greedy algorithm as the oracle, it can be implemented in polynomial time, and is a $4e/(e-1)$-approximation mechanism.
\end{theorem}

For calculating the agents' threshold payments of our oracle mechanisms in this section, it is enough to check what is the maximum cost that each agent $i$ can declare such that she is still in the winning set. Similar to section \ref{Sec:Param1}, for each agent $i$, this number can simply be computed by checking where in the sorted list of agents by their marginal bang-per-buck this agent can appear such that the inequality of \textsc{Greedy-EOM}$(\alpha,A,B)$ (for the exponential time mechanisms) and the inequality of \textsc{Greedy-OM}$(\alpha,A,B)$ (for polynomial time mechanisms) still hold. The characterization of these threshold payments is similar to the payment characterization of the oracle mechanisms in \cite{anari2014}. 
\label{Sec:Param2}

\section{A Simple \lowercase{$1+\frac{e}{e-1}$} Approximation mechanism for Large Markets}
\label{Sec:LargeMarkets}
In this section we combine the two greedy parameterized mechanisms of Section \ref{Sec:Param}, \textsc{Greedy-OM}$(\alpha,A,B)$ and \textsc{Greedy-TM}$(\gamma,A,B)$ to improve the approximation guarantee for large markets. 
Given a polynomial time $r$ approximation oracle, our simple, deterministic, truthful, and budget feasible mechanism in this section has an approximation ratio of $1+r$ and runs in polynomial time.

\begin{wrapfigure}[6]{r}{2.4in}
\IncMargin{2em}
\vspace{-0.15in}
\begin{algorithm}[H]
\label{Mech:Large}
\DontPrintSemicolon
\textsc{Deterministic-Large}$(\alpha,\gamma, A, B)$\;
Let $A= \{i:c_i\leq B\}$\;
Let $S_{\alpha}= $ \textsc{Greedy-OM}$(\alpha,A,B)$\;
Let $S_{\gamma}= $ \textsc{Greedy-TM}$(\gamma,A,B)$\;
\Return $S_{\alpha} \cap S_{\gamma}$
\end{algorithm}
\IncMargin{-2em}
\end{wrapfigure}

At first glance, \textsc{Deterministic-Large} seems worse than both of \textsc{Greedy-OM} and \textsc{Greedy-Tm}, since its winning set is the intersection of the wining sets of these mechanisms. However, taking the intersection of these mechanisms will allow us to choose the value of $\alpha$ and $\gamma$ to be higher than $0.5$ while keeping the mechanism budget feasible.\pjcomment{Is there a way to make the box be in the right and the text on the left?}

First, note that the intersection of two monotone mechanisms is monotone.

\begin{proposition}\label{Prop:intersection}
For two monotone mechanisms $M_1$ and $M_2$, the mechanism $M$ that outputs the intersection of the winning set of $M_1$ and the winning set of $M_2$ is monotone.
\end{proposition}
\begin{proof}
Assume an agent $i\in A$ is not in the winning set of $M$. This is because either $i$ is not in the winning set of $M_1$ or not in the winning set of $M_2$ (or both). W.l.o.g assume that $i$ is not in the winning set of $M_1$. Since $M_1$ is monotone, if $i$ increases his bid, he still will not be in the winning set of $M_1$. This means that after $i$ increases his bid, he is still not in the winning set of $M$. Therefore, $M$ is monotone.
\end{proof}

Next we give a parameterized approximation guarantee for \textsc{Deterministic-Large}$(\alpha,\gamma, A, B)$.

\begin{lemma}\label{Lem:LargeApprox}
Assuming the large market assumption, for every fixed value of $\alpha, \gamma \in (0,1]$ \textsc{Deterministic-Large}$(\alpha,\gamma,A,B)$ is monotone. Furthermore, with an $r$ approximation oracle, it has a worst case approximation ratio of $\max(1+\frac{1}{\gamma},\frac{r}{\alpha})$.
\end{lemma}
\begin{proof}
From Proposition \ref{Prop:intersection}, Lemmas \ref{Lem:GammaApprox}, \ref{Lem:AlphaAprox} it follows that \textsc{Deterministic-Large} is monotone.

Let $S$ be the outcome of the mechanism. Let $k$ be the biggest integer such that $S_{k-1} \subseteq S$, i.e., $S_{k-1} \subseteq S_{\alpha}$ and $S_{k-1} \subseteq S_{\gamma}$. By definition $k \notin S$, so there are two cases
\begin{itemize}
\item $k \notin S_{\alpha}$: By lemma \ref{Lem:AlphaAprox}, the large market assumption and monotonicity of $v(.)$, we have $\frac{r}{\alpha} v(S) \approx \frac{r}{\alpha}v(S) + (1+\frac{r}{\alpha}) v(i^*) \geq \frac{r}{\alpha}v(S_{k-1})\etedit{ +(1+\frac{r}{\alpha}) v(i^*)} \geq opt(A,B)$.

\item $k \notin S_{\gamma}$: By lemma \ref{Lem:GammaApprox}, the large market assumption and monotonicity of $v(.)$, we have $(1+\frac{1}{\gamma})v(S_{k-1}) \approx (1+\frac{1}{\gamma})v(S_{k-1}) + \frac{1}{\gamma}v(i^*) \geq opt(A,B)$.
\end{itemize}

In both cases we have, $\max(1+\frac{1}{\gamma},\frac{r}{\alpha})v(S) \geq opt(A,B)$ assuming that $v(i^*)$ is negligible.
\end{proof}

Now we provide a simple condition for the budget feasibility of \textsc{Deterministic-Large}$(\alpha,\gamma, A, B)$.

\begin{lemma}\label{Lem:LargeBudget}
If $\alpha \leq \frac{1}{1+\gamma}$ for any $\alpha, \gamma \ge 0$, then by using threshold payments, \textsc{Deterministic-Large}$(\alpha,\gamma, A, B)$ is budget feasible.
\end{lemma}
\begin{proof}
Let $p_i$ be the threshold payment for agent $i$. Let $S=$\textsc{Deterministic-Large}$(\alpha,\gamma, A, B)$. For every $i \in S$, we show that if $i$ deviates to bidding a cost $b_{i} > m_{i}(S_{i-1})\frac{B}{v(S)}$, he cannot be in the winning set. By proving this and by using the definition of threshold payments we get
$\sum_{i \in S}p_{i} \leq \sum_{i \in S}m_{i}(S_{i-1}) \frac{B}{v(S)}\leq \sum_{i \in S}m_{i}(S_{i-1}\cap S) \frac{B}{v(S)} = B$, so the mechanism is budget feasible.

We prove above claim by contradiction: assume that $i$ deviates to $b_{i} >m_{i}(S_{i-1})\frac{B}{v(S)}$ and is in the winning set. Let $b$ be the new cost vector and $j$ be position of $i$ in the new order of items. Let $S'_z$ for $z\in [n]$ be defined similar to $S_z$ but with cost vector $b$ instead of $c$. So $S'_j$ is the set of items that are in the winning set of \textsc{Greedy-TM}$(\gamma,A,B)$ at the end of step $j$ once $i$ is added. Note that $S'_j$ is also equal to the set of all the items that has been considered by \textsc{Greedy-OM}$(\alpha,A,B)$ at the end of its $j$-th step. So by using the same argument as proof of lemma \ref{Lem:AlphaBudget} we get
$$c(S^*)>B\frac{v(S^*)-v(S'_j)}{v(S)} \mbox{\hspace{0.5in} and \hspace{0.5in}} v(S'_j),v(S) \leq \alpha v(S^*)$$
By defining $x = \frac{v(S'_j)}{v(S)}$, we have $v(S'_j)=xv(S) \leq \alpha xv(S^*)$. So we get
\begin{align*}
c(S^*)&>B\frac{v(S^*)-v(S'_j))}{v(S)}
>B\frac{(1-\alpha)v(S^*)}{\alpha xv(S^*)} = B \frac{1-\alpha}{\alpha x}
\end{align*}
so if $1-\alpha \geq \alpha x$, or equivalently, $\alpha \leq \frac{1}{1+x}$ then we get $c(S^*)>B$ which is the desired contradiction.

Since $i \in S_{\gamma}$, we also have
\begin{align*}
\frac{b_{i}}{m_{i}(S'_j)} &\leq \gamma \frac{B}{v(S'_j)} = \frac{\gamma}{x} \frac{B}{v(S)}
\end{align*}
So since $m_{i}(S'_j) \leq m_{i}(S_{i-1})$, if $\gamma \leq x$, we get to a contradiction with the assumption about $b_i$.

The only remaining case is when $\gamma > x$ and $\alpha > \frac{1}{1+x}$. This means that $\alpha > \frac{1}{1+ \gamma}$ which is a contradiction with the property in the statement of lemma, so the mechanism is budget feasible.
\end{proof}

By using Lemmas \ref{Lem:LargeApprox}, \ref{Lem:LargeBudget}, the main theorem of this section follows.

\begin{theorem}\label{THM:LargeMain}
By using threshold payments, \textsc{Deterministic-Large} $(\frac{r}{r+1},\frac{1}{r})$ is truthful, budget feasible, and $1+r$ approximation of the optimum. By using the greedy algorithm with $r=e/(e-1)$ we get a mechanism with approximation guarantee of $\approx 2.58$.
\end{theorem}
\begin{proof}
By lemma \ref{Lem:LargeApprox} we know that the mechanism is monotone, so by using threshold payments, the mechanism is truthful.

Let $S$ be the winning set of the mechanism. By lemma \ref{Lem:LargeApprox}, we have
\begin{align*}
&\max(1+\frac{1}{\gamma},\frac{r}{\alpha})v(S)\geq opt(A,B)\\
\Rightarrow &\max(1+r, \frac{r}{\frac{r}{r+1}})v(S) \geq opt(A,B)\\
\Rightarrow &(1+r)v(S) \geq opt(A,B)
\end{align*}
So the mechanism is $1+r$ approximation of the optimum.

We also have that
\begin{align*}
\alpha = \frac{r}{r+1} = \frac{\frac{1}{\gamma}}{\frac{\gamma+1}{\gamma}} = \frac{1}{1+\gamma}
\end{align*}
So by lemma \ref{Lem:LargeBudget} the mechanism is budget feasible.

By choosing the oracle to be the greedy algorithm in \cite{greedy2004}, which is $\frac{e}{e-1}$ approximation of the optimum, the approximation ratio of \textsc{Deterministic-Large} is $1+\frac{e}{e-1}\approx 2.58$.
\end{proof}

In order to calculate the threshold payment of the agents in winning set of this mechanism, it is enough to calculate the threshold payment of each agent in the threshold mechanism and the oracle mechanism inside \textsc{Deterministic-Large}$(\alpha,\gamma, A, B)$, and declare their minimum as the payment of the agent.

\section{Application to hiring in Crowdsourcing Markets
}\label{Sec:Application}
In this section, we consider an application of our knapsack problem with heterogeneous items to the problem of a principal hiring in a Crowdsourcing Market. We consider the model where there is a set of agents $A$ that can be hired and a set of tasks $T$  that the principal would like to get get done. Each agent $a\in A$ has a private cost. We represent the abilities of the agents by a bipartite graph $G(A,T)$, where edge $e=(a,t)$ in the $G$ indicates that agent $a$ can be used for task $t$. The value of buyer for each edge $e$ is $v_e$, which can be different for each edge. The principal has a budget $B$, and would like to hire agents with cost under her budget to maximize the total value of the task done. The optimal solution for this problem is a maximum value matching, subject to the budget constraint on the cost of the agents hired.

Such knapsack with heterogeneous items and agents with matching constraints were studied by \cite{singer2010,chen2011,goel2014}. There are many ways for modeling the heterogeneity of items.
In \cite{chen2011}, this heterogeneity has been defined by having types for items where at most one item can be chosen from each type, corresponding to a bipartite graph where agents have degree 1. \cite{goel2014} consider our model of agents and tasks with a bipartite graph, but assume that the principal has a fixed value for each task completed, independent of the agent that took care of the task, so values of the edges entering a task node $t$ are all equal.

In this section, we apply our technique from section \ref{Sec:LargeMarkets} for this problem. We first consider the general problem defied above, but relax the assumption that the allocation should always assigned each agent in the winning set to a unique task. Turns out that allowing the principal to hire extra agents simplifies the problem. We define the value of the buyer for the winning set $S$ to be the value of the maximum matching on the induced subgraph $G[S,T]$. Then we consider the special case of \cite{goel2014}, where the principal has a value for each task, independent of who completes the task. For this case, we show that a small change in our mechanism (stopping it when the marginal increase in value is 0) results in the same approximation guarantee while observing the hard constraint that all agents hired need to be assigned to a task.

\paragraph{\textbf{General Crowdsourcing Markets}}
The following proposition states that the maximum value of a matching in the induced subgraph $G[S,T]$ for a subset of agents $S$  is a monotone submodular function of $S$. The monotonicity proof follows directly from the definition. It is not hard to prove that the function is also submodular (see \cite{StackExchange} ).

\begin{proposition}\label{Prop:SubmodularMatching}
For $S \subseteq A$, if $f(S)$ is the value of maximum weight matching of the induced subgraph $G[S,T]$ of the weighted bipartite graph $G(A,T)$, then $f(S)$ is a monotone submodular function.
\end{proposition}

This proposition implies that all our truthful budget feasible mechanisms for submodular valuations can be used for this model.

\begin{corollary}
Any budget feasible truthful mechanism for submodular valuations, can also be used without any loss in the approximation guarantee for case of heterogeneous tasks (items), if it is OK to have agents in the winning set who are not assigned to any tasks.
\end{corollary}

\paragraph{\textbf{Hiring with strict matching constraint.}}
Now, consider the case where the buyer's value is defined by summation of her value for each task, i.e. for all the edges that are directly connected to the same task, the value of the buyer for those edge is the same. We argue that in this model, if we add the hard constraint each agent in the winning set should be assigned to a unique task, then with a small change in our mechanisms, all our results still hold. This problem was considered by \cite{goel2014} for large markets, who gave a randomized truthful (in expectation) and budget feasible mechanism with a $1+\frac{e}{e-1}$ approximation guarantee for large markets (the main result of \cite{goel2014}). Next lemma shows how one can use our deterministic truthful budget feasible mechanism for large markets to get the same approximation guarantee.

\begin{lemma}\label{Lem:KnapsackMatching}
For $S \subseteq A$, let $f(S)$ be the value of maximum weighted matching of the induced subgraph $G[S,T]$ of the weighted bipartite graph $G(A,T)$ in which all the edges that connect to the same node of $T$ has the same value. If a maximum weight matching induced by $S \subseteq A$ connects all vertices in $S$ to a vertex in $T$, and for $a \in A \setminus S$, $f(S\cup \{a\})-f(S)>0$, then there is a maximum weight matching induced by $S \cup \{a\}$ which is also assigning each agent to a unique task.
\end{lemma}
\begin{proof}
We use contradiction. Assume that there is a subset of agents $S\subseteq A$ such that there is a maximum matching $M$ in the subgraph of $G$, induced by vertices of $S$ and $T$ that connects each agent in $S$ to a task in $T$. Let $a \in A$ be an agent such that $f(S\cup \{a\})-f(S)>0$ and there is no maximum matching in the subgraph induced by $S'=S \cup \{a\}$ and $T$ that connects each agent in $S'$ to a task in $T$. Let $M'$ be a maximum matching of this induced subgraph. Let $G'$ be the union of edges in $M$ and $M'$ and let $C$ and $P$ be the set of cycles and paths that contain all the edges of $G'$. Since $M$ and $M'$ are both maximum matchings, we have $W(M\cap c)=W(M'\cap c)$ for all $c \in C$. Since the only difference between $S$ and $S'$ is having $a$, there can only be one path $p \in P$ such that $W(p \cap M')>W(p \cap M)$. Furthermore, one of the end points of $p$ should be $a$ and for all other paths $p' \in P$ that $p'\neq p$, $W(p' \cap M) = W(p' \cap M')$. For $p$ there are two cases
\begin{itemize}
\item The edge that is connected to the other endpoint of $p$ is in $M$: in this case, since the value of matching is defined by tasks, $W(p\cap M)=W(p \cap M')$.  Therefore, $W(M)=W(M')$ and $F(S')-F(S)=0$ which is a contradiction.
\item The edge that is connected to the other endpoint of $p$ is in $M'$: In this case if we define a matching $M^*=(M \setminus p) \cup (M' \cap p)$, then $M^*$ will connect each agent in $S'$ to a unique task, which is a contradiction.
\end{itemize}
This means that we reach contradiction in both cases, and the proof is complete.
\end{proof}

For using lemma \ref{Lem:KnapsackMatching} in our mechanisms, we only have to stop considering items in the sorted list of marginal bang-per-bucks whenever the marginal bang-per-buck of the item is 0. Note that since the items are listed in decreasing order of marginal bang-per-buck and we know that the valuation is submodular, doing this will not have any effects on the approximation ratio (since the marginal bang-per-buck of the next items is also 0) and truthfulness (since the threshold payment of an agent whose item has 0 marginal value is 0) of our mechanisms. The following corollary summarizes this result.

\begin{corollary}
The truthful budget feasible threshold and oracle mechanisms, as well as the large markets mechanism for submodular valuations of this paper without any loss in the approximation ratio can be also used for the case of heterogeneous tasks, with the constraint that each agent in the winning set should be assigned to a unique task (matching constraint).
\end{corollary}

For the case of large markets, using Theorem \ref{THM:LargeMain} of section \ref{Sec:LargeMarkets} gives a deterministic truthful and budget feasible mechanism for this problem, matching the $1+\frac{e}{e-1}$ guarantee of the randomized truthful (in expectation) mechanism of \cite{goel2014}, while keeping the mechanism deterministic. 

\bibliographystyle{alpha}
\bibliography{BibFile}
\newpage
\appendix
\section{Deferred Proof of Theorem 3.4}
\label{Sec:Appendix}
Recall \textbf{Theorem \ref{THM:Chen}}, which states that \textsc{Random-TM}$(0.5,A,B)$ and \textsc{Greedy-TM}$(0.5,A,B)$ are budget feasible.

\begin{proof}
In order to prove \textsc{Random-TM}$(0.5,A,B)$ is budget feasible, we only have to show \textsc{Greedy-TM}$(0.5,A,B)$ is budget feasible, since when $i^*$ is selected, his threshold payment is $B$ and the mechanism is budget feasible.

Let $p_i$ be the threshold payment for agent $i$. Let $S_{k-1}=$\textsc{Greedy-TM}$(0.5,A,B)$ the set returned by the greedy threshold mechanism. For every $i \in S_{k-1}$, we show that if $i$ deviates to bidding a cost $b_i > m_i(S_{i-1})\frac{B}{v(S_{k-1})}$, then he would not be selected. This will imply that the threshold payment $p_i$ for player $i$ is at most $p_i\le m_i(S_{i-1})\frac{B}{v(S_{k-1})}$, and so we get
$\sum_{i \in S_{k-1}}p_i \leq \sum_{i \in S_{k-1}}m_i(S_{i-1}) \frac{B}{v(S_{k-1})}= B$, so the mechanism is budget feasible.

Consider the run of \textsc{Greedy-TM} where $i$ deviated to bidding $b_i>m_i(S_{i-1})\frac{B}{v(S_{k-1})}$ and is selected, while all other players bid truthfully. Let $b$ denote the resulting cost-vector. Note that by bidding higher, player $i$ would occur later in the order.
Let $j$ be the step in which $i$ would occur, after he deviates to bidding cost $b_i$. Let $S'_j$ be the items that are in the winning set at the end of this step ($S'_z$ for $z\in [n]$ is defined similar to $S_z$ but with cost vector $b$ instead of $c$, with the change in one cost also effecting the order of items after item $i$). If $i$ is in the winning set, we have
\begin{align*}
\frac{b_i}{m_i(S'_{j-1})} &\leq 0.5\frac{B}{v(S'_{j-1})}
\end{align*}
Since the items are sorted by their marginal bang per buck, for every $z \in S_{k-1}$, $c_z \leq 0.5 B\frac{m_z(S_{z-1})}{v(S_{k-1})}$, so we have $c(S_{k-1}) \leq 0.5B$. Let $T=S_{k-1} \setminus S'_j=\{t_1,t_2,\ldots,t_q\}$, $T_0=\emptyset$ and $T_z=\{t_l:l\in[z]\}$. So we have
\begin{align*}
v(S_{k-1})-v(S'_{j}) \leq v(S_{k-1}\cup S'_{j}) - v(S'_{j}) = \sum_{z \in [q]} m_{t_z}(S'_{j} \cup T_{z-1}) = \sum_{z \in [q]} c_{t_z} \frac{m_{t_z}(S'_{j} \cup T_{z-1})}{c_{t_z}}
\end{align*}
By submodularity we have
\begin{align*}
\sum_{z \in [q]} c_{t_z} \frac{m_{t_z}(S'_{j} \cup T_{z-1})}{c_{t_z}}&\leq \sum_{z \in [q]} c_{t_z} \frac{m_{t_z}(S'_{j-1})}{c_{t_z}}
\end{align*}
Since $i$ is selected at step $j$ it means it has the highest marginal bang per buck at that step. The cost vectors $b$ and $c$ are also only different in the cost of $i$. So we have
\begin{align*}
\sum_{z \in [q]} c_{t_z} \frac{m_{t_z}(S'_{j-1})}{c_{t_z}} \leq \sum_{z \in [q]} c_{t_z} \frac{m_i(S'_{j-1})}{b_i} \leq \frac{m_i(S'_{j-1})}{b_i} \sum_{z \in [q]} c_{t_z}
\end{align*}
Since $i$ has increased his cost, he cannot be selected before step $i$, so $S_{i} \subseteq S'_{j}$ and $v(S_{i})\leq v(S'_j)$. By using this and contradiction assumption we have
\begin{align*}
\frac{m_i(S'_{j-1})}{b_i} \sum_{z \in [q]} c_{t_z} < \frac{v(S_{k-1})}{B} \frac{B}{2} =\frac{v(S_{k-1})}{2}
\end{align*}
So $v(S_{k-1}) \leq 2v(S'_j)$. By adding this to the previous inequality and replacing $m_i(S'_{j-1})$ with $m_i(S_{i-1})$  (note that we can do this since $S_{i-1} \subseteq S'_{j-1}$), we get to a contradiction.
\end{proof}




\end{document}